\documentclass[11pt,letterpaper]{article}

\usepackage{amsmath,amsfonts,amsthm,amssymb,stmaryrd,graphicx,relsize,bm}
\theoremstyle{plain}

\numberwithin{equation}{section}

\newtheorem{thm}{Theorem}[section]
\newtheorem{lem}[thm]{Lemma}

\newenvironment{exam}[1]
{\begin{flushleft}\textbf{Example #1}.\enspace}%
{\end{flushleft}}

\allowdisplaybreaks  

\newcommand{\complex}{{\mathbb C}}
\newcommand{\real}{{\mathbb R}}

\newcommand{\tbullet}{\mathrel{\raise .4ex\hbox{\tiny$\bullet$}}} 

\newcommand{\rmtr}{\mathrm{tr\,}}
\newcommand{\rmcor}{\mathrm{Cor}}
\newcommand{\rmre}{\mathrm{Re\,}}
\newcommand{\rmim}{\mathrm{Im\,}}
\newcommand{\ityes}{\textit{yes}}
\newcommand{\itno}{\textit{no}}

\newcommand{\hscript}{\mathcal{H}}
\newcommand{\iscript}{\mathcal{I}}
\newcommand{\lscript}{\mathcal{L}}
\newcommand{\oscript}{\mathcal{O}}
\newcommand{\rscript}{\mathcal{R}}
\newcommand{\sscript}{\mathcal{S}}
\newcommand{\tscript}{\mathcal{T}}

\newcommand{\ahat}{\widehat{A}}
\newcommand{\bhat}{\widehat{B}}
\newcommand{\atilde}{\widetilde{A}}

\newcommand{\btilde}{\widetilde{B}}
\newcommand{\iscripttilde}{\widetilde{\iscript}}

\newcommand{\iscriptbar}{\overline{\iscript}}

\newcommand{\rarrow}{\overrightarrow{r}}
\newcommand{\sigmaarrow}{\overrightarrow{\sigma}}

\newcommand{\ab}[1]{\left|#1\right|}
\newcommand{\doubleab}[1]{\left|\left|#1\right|\right|}
\newcommand{\brac}[1]{\left\{#1\right\}}
\newcommand{\paren}[1]{\left(#1\right)}
\newcommand{\sqbrac}[1]{\left[#1\right]}
\newcommand{\elbows}[1]{{\left\langle#1\right\rangle}}
\newcommand{\ket}[1]{{\left|#1\right>}}
\newcommand{\bra}[1]{{\left<#1\right|}}

\begin{document}

\title{REAL\,-VALUED OBSERVABLES\\ AND QUANTUM UNCERTAINTY}
\author{Stan Gudder\\ Department of Mathematics\\
University of Denver\\ Denver, Colorado 80208\\
sgudder@du.edu}
\date{}
\maketitle

\bigskip
\textit{Dedicated to the memory of Richard Greechie (1941--2022).}

\textit{The author's cherished friend, long time colleague and collaborator.}
\bigskip

\begin{abstract}
We first present a generalization of the Robertson-Heisenberg uncertainty principle. This generalization applies to mixed states and contains a covariance term. For faithful states, we characterize when the uncertainty inequality is an equality. We next present an uncertainty principle version for real-valued observables. Sharp versions and conjugates of real-valued observables are considered. The theory is illustrated with examples of dichotomic observables. We close with a discussion of real-valued coarse graining.
\end{abstract}

\section{Introduction}  
One of the basic principles of quantum theory is the Robertson-Heisenberg uncertainty inequality \cite{hz12,nc00}
\begin{equation}                
\label{eq11}
\Delta _\psi (A)\Delta _\psi (B)\ge\tfrac{1}{4}\ab{\elbows{\psi ,\sqbrac{A,B}\psi}}^2
\end{equation}
where $A,B$ are self-adjoint operators and $\psi$ is a vector state on a Hilbert space. The inequality \eqref{eq11} is usually applied to position and momentum operators $A,B$ in which case $\ab{\elbows{\psi ,\sqbrac{A,B}\psi}}^2=\hbar ^2$ where $\hbar$ is Planck's constant. In this situation, $A$ and $B$ are unbounded operators, but for mathematical rigor we shall only deal with bounded operators. However, our results can be extended to the unbounded case by considering a dense subspace common to the domains of $A$ and $B$. In this paper, we derive a generalization of \eqref{eq11}. This generalization applies to mixed states and contains an additional covariance term that results in a stronger inequality.

The main result in Section~2 is an uncertainty principle for observable operators. This principle contains four parts: a commutator term, a covariance term, a correlation term and a product of variances term. This last term is sometimes called a product of uncertainties. In Section~2 we also characterize, for faithful states, when the uncertainty inequality is an equality. Section~3 introduces the concept of a real-valued observable. If $\rho$ is a state and $A$ is a real-valued observable, we define the $\rho$-average,
$\rho$-deviation and $\rho$-variance of $A$. If $B$ is another real-valued observable, we define the $\rho$-correlation and
$\rho$-covariance of $A,B$. An uncertainty principle for real-valued observables is given in terms of these concepts. An important role is played by the stochastic operator $\atilde$ for $A$. In Section~3 we also define the sharp version of a real-valued observable and characterize when two real-valued observables have the same sharp version

Section~4 illustrates the theory presented in Section~3 with two examples. The first example considers two dichotomic arbitrary real-valued observables. The second example considers the special case of two noisy spin observables. In this case, the uncertainty inequality becomes very simple. Section~5 discusses real-values coarse graining of observables.

\section{Quantum Uncertainty Principle}  
For a complex Hilbert space $H$, we denote the set of bounded linear operators by $\lscript (H)$ and the set of bounded self-adjoint operators by $\lscript _S(H)$. A positive trace-class operator with trace one is a \textit{state} and the set of states on $H$ is denoted by $\sscript (H)$. A state $\rho$ is \textit{faithful} if $\rmtr (\rho C^*C)=0$ for $C\in\lscript (H)$ implies that $C=0$. For 
$\rho\in\sscript (H)$ and $C,D\in\lscript (H)$ we define the sesquilinear form $\elbows{C,D}_\rho =\rmtr (\rho C^*D)$.

\begin{lem}    
\label{lem21}
{\rm{(i)}}\enspace If $C\in\lscript (H)$, $\rho\in\sscript (H)$, then $\rmtr (\rho C^*)=\overline{\rmtr (\rho C)}$.
{\rm{(ii)}}\enspace The form $\elbows{\tbullet ,\tbullet}_\rho$ is a positive semi-definite inner product.
{\rm{(iii)}}\enspace A state $\rho$ is faithful if and only if $\elbows{\tbullet ,\tbullet}_\rho$ is an inner product
\end{lem}
\begin{proof}
(i)\enspace If $D$ is a trace-class operator and $\brac{\phi _i}$ is an orthonormal basis for $H$, we have
\begin{equation*}
\rmtr (D^*)=\sum _i\elbows{\phi _i,D^*\phi _i}=\sum _i\overline{\elbows{D^*\phi _i,\phi _i}}=\sum _i\overline{\elbows{\phi _i,D\phi _i}}
   =\overline{\rmtr (D)}
\end{equation*}
Hence,
\begin{equation*}
\rmtr (\rho C^*)=\rmtr\sqbrac{(C\rho )^*}=\overline{\rmtr (C\rho )}=\overline{\rmtr (\rho C)}
\end{equation*}
(ii)\enspace Applying (i), we have
\begin{equation*}
\overline{\elbows{C,D}_\rho}=\overline{\rmtr (\rho C^*D)}=\rmtr\sqbrac{\rho (C^*D)^*}=\rmtr (\rho D^*C)=\elbows{D,C}_\rho
\end{equation*}
Moreover, since $C^*C\ge 0$ we have $\elbows{C,C}_\rho =\rmtr (\rho C^*C)\ge 0$. Hence, $\elbows{\tbullet ,\tbullet}_\rho$ is a positive semi-definite inner product.
(iii)\enspace If $\elbows{\tbullet ,\tbullet}_\rho$ is an inner product, then
\begin{equation*}
\elbows{C,C}_\rho =\rmtr (\rho C^*C)=0
\end{equation*}
implies $C=0$ so $\rho$ is faithful. Conversely, if $\rho$ is faithful, then
\begin{equation*}
\rmtr (\rho C^*C)=\elbows{C,C}_\rho =0
\end{equation*}
implies $C=0$ so $\elbows{\tbullet ,\tbullet}_\rho$ is an inner product
\end{proof}

For $A\in\lscript _S(H)$ and $\rho\in\sscript (H)$, the $\rho$-\textit{average} (or $\rho$-\textit{expectation}) of $A$ is
$\elbows{A}_\rho=\rmtr (\rho A)$ and $\rho$-\textit{deviation} of $A$ is $D_\rho (A)=A-\elbows{A}_\rho I$ where $I$ is the identity map on $H$. If $A,B\in\lscript _S(H)$, the $\rho$-\textit{correlation} of $A,B$ is
\begin{equation*}
\rmcor _\rho (A,B)=\rmtr\sqbrac{\rho D_\rho (A)D_\rho (B)}
\end{equation*}
Although $\rmcor _\rho (A,B)$ need not be a real number, it is easy to check that $\overline{\rmcor _\rho (A,B)}=\rmcor _\rho (B,A)$. We say that $A$ and $B$ are \textit{uncorrelated} if $\rmcor _\rho (A,B)=0$. The $\rho$-\textit{covariance} of $A,B$ is
$\Delta _\rho (A,B)=\rmre\rmcor _\rho (A,B)$ and the $\rho$-\textit{variance} of $A$ is
\begin{equation*}
\Delta _\rho (A)=\Delta _\rho (A,A)=\rmcor _\rho (A,A)=\rmtr\sqbrac{\rho D_\rho (A)^2}
\end{equation*}
It is straightforward to show that
\begin{align}             
\label{eq21}
\rmcor _\rho (A,B)&=\rmtr (\rho AB)-\elbows{A}_\rho\elbows{B}_\rho\\
\label{eq22}
\Delta _\rho (A,B)&=\rmre\rmtr (\rho AB)-\elbows{A}_\rho\elbows{B}_\rho\\
\label{eq23}
\Delta _\rho (A)&=\elbows{A^2}_\rho-\elbows{A}_\rho ^2
\end{align}
We see from \eqref{eq21} that $A$ and $B$ are $\rho$-uncorrelated if and only if $\rmtr (\rho AB)=\elbows{A}_\rho\elbows{B}_\rho$. We say that $A$ and $B$ \textit{commute} if their commutant $\sqbrac{A,B}=AB-BA=0$.

\begin{exam}{1}  
In the tensor product $H_1\otimes H_2$ let $\rho =\rho _1\otimes\rho _2\in\sscript (H_1\otimes H_2)$ be a product state and let $A_1\in\lscript _S(H_1)$, $A_2\in\lscript _S(H_2)$. Then $A=A_1\otimes I_2$, $B=I_1\otimes A_2\in\lscript _S(H_1\otimes H_2)$ are $\rho$-uncorrelated because
\begin{align*}
\rmtr (\rho AB)&=\rmtr\sqbrac{\rho _1\otimes\rho _2(A_1\otimes I_2)(I_2\otimes A_2)}
   =\rmtr\sqbrac{\rho _1\otimes\rho _2(A_1\otimes A_2)}\\
   &=\rmtr (\rho _1A_1\otimes\rho _2A_2)=\rmtr (\rho _1A_1)\rmtr (\rho _2A_2)\\
   &=\rmtr (\rho _1\otimes\rho _2A_1\otimes I_2)\rmtr (\rho _1\otimes\rho _2I_1\otimes A_2)=\elbows{A}_\rho\elbows{B}_\rho
\end{align*}
This shows that $A,B$ are $\rho$-uncorrelated for any product state $\rho$. Of course, $\sqbrac{A,B}=0$ in this case. However, there are examples of noncommuting operators that are uncorrelated. For instance, on $H=\complex ^2$ let
$\alpha =\begin{bmatrix}1\\0\end{bmatrix}$, $\phi =\begin{bmatrix}0\\1\end{bmatrix}$,
$\psi =\tfrac{1}{\sqrt{2}}\begin{bmatrix}1\\1\end{bmatrix}$. With $\rho =\ket{\alpha}\bra{\alpha}$, $A=\ket{\phi}\bra{\phi}$,
$B=\ket{\psi}\bra{\psi}$ we have
\begin{equation*}
\rmtr (\rho AB)=\elbows{A}_\rho\elbows{B}_\rho =0
\end{equation*}
Hence, $A,B$ are $\rho$-uncorrelated. However,
\begin{align*}
AB&=\elbows{\phi ,\psi}\ket{\phi}\bra{\psi}=\tfrac{1}{\sqrt{2}}\ket{\phi}\bra{\psi}\\
BA&=\elbows{\psi ,\phi}\ket{\psi}\bra{\phi}=\tfrac{1}{\sqrt{2}}\ket{\psi}\bra{\phi}
\end{align*}
so $\sqbrac{A,B}\ne 0$.\hfill\qedsymbol
\end{exam}

We now present our main result.

\begin{thm}    
\label{thm22}
If $A,B\in\lscript _S(H)$ and $\rho\in\sscript (H)$, then
{\rm{(i)}}\enspace $\tfrac{1}{4}\ab{\rmtr\paren{\rho \sqbrac{A,B}}}^2+\sqbrac{\Delta _\rho (A,B)}^2=\ab{\rmcor _\rho (A,B)}^2$
\newline
{\rm{(ii)}}\enspace $\tfrac{1}{4}\ab{\rmtr\paren{\rho \sqbrac{A,B}}}^2+\sqbrac{\Delta _\rho (A,B)}^2\le\Delta _\rho (A)\Delta _\rho (B)$
\end{thm}
\begin{proof}
(i)\enspace Applying Lemma~\ref{lem21} we have
\begin{align}             
\label{eq24}
\rmtr\paren{\sqbrac{A,B}}&=\rmtr (\rho AB)-\rmtr (\rho BA)=\rmtr (\rho AB)-\overline{\rmtr\sqbrac{\rho (BA)^*}}\notag\\
  &=\rmtr (\rho AB)-\overline{\rmtr(\rho A^*B^*)}=\rmtr (\rho AB)-\overline{\rmtr (\rho AB)}\notag\\
  &=2i\,\rmim\sqbrac{\rmtr (\rho AB)}
\end{align}
From \eqref{eq22} and \eqref{eq24} we obtain
\begin{align*}
\tfrac{1}{4}\ab{\rmtr\paren{\rho\sqbrac{A,B}}}^2+\sqbrac{\Delta _\rho (A,B)}^2
   &=\sqbrac{\rmim (\rho AB)}^2+\sqbrac{\rmre\rmtr (\rho A B)-\elbows{A}_\rho\elbows{B}_\rho}^2\\
   &=\ab{\rmre\rmtr (\rho AB)-\elbows{A}_\rho\elbows{B}_\rho +i\,\rmim\rmtr (\rho AB)}^2\\
   &=\ab{\rmtr (\rho AB)-\elbows{A})_\rho\elbows{B}_\rho}^2=\ab{\rmcor _\rho (A,B)}^2
\end{align*}
(ii)\enspace Applying Lemma~\ref{lem21}(ii), the form $\elbows{C,D}_\rho=\rmtr(\rho C^*D)$ is a positive semi-definite inner product. Hence, Schwarz's inequality holds and we have
\begin{align*}
\ab{\rmcor _\rho (A,B)}^2&=\ab{\rmtr\sqbrac{\rho D_\rho (A)D\rho (B)}}^2=\ab{\elbows{D_\rho (A),D_\rho (B)}_\rho}^2\\
   &\le\elbows{D_\rho (A),D_\rho (A)}_\rho\elbows{D_\rho (B),D_\rho (B)}_\rho
   =\rmtr\sqbrac{\rho D_\rho (A)^2}\rmtr\sqbrac{\rho D_\rho (B)^2}\\
   &=\Delta _\rho (A)\Delta _\rho (B)\qedhere
\end{align*}
\end{proof}

We call Theorem~\ref{thm22}(i) the \textit{uncertainty equation} and Theorem~\ref{thm22}(ii) the \textit{uncertainty inequality}. Together, they are called the \textit{uncertainty principle}. Notice that Theorem~\ref{thm22}(ii) is a considerable strengthening of the usual Robertson-Heisenberg inequality \eqref{eq11} since it contains the term $\sqbrac{\Delta _\rho (A,B)}^2$ and it applies to arbitrary states. Thus, even when $\sqbrac{A,B}=0$ we still have an uncertainty relation
\begin{equation*}
\sqbrac{\Delta _\rho (A,B)}^2=\ab{\rmtr\sqbrac{\rho \Delta _\rho (A)\Delta _\rho (B)}}^2\le\Delta _\rho (A)\Delta _\rho (B)
\end{equation*}

\begin{lem}    
\label{lem23}
A state $\rho$ is faithful if and only if the eigenvalues of $\rho$ are positive.
\end{lem}
\begin{proof}
Suppose the eigenvalues $\lambda _i$ of $\rho$ are positive with corresponding normalized eigenvectors $\phi _i$. Then we can write $\rho =\sum\lambda _i\ket{\phi _i}\bra{\phi _i}$ for the orthonormal basis $\brac{\phi _i}$. For any $A\in\lscript (H)$ we obtain
\begin{equation*}
\rmtr (\rho A^*A)=\sum\lambda _i\rmtr\paren{\ket{\phi _i}\bra{\phi _i}A^*A}=\sum\lambda _i\elbows{A\phi _i,A\phi _i}
   =\sum\lambda _i\doubleab{A\phi _i}^2
\end{equation*}
Hence, $\rmtr (\rho A^*A)=0$ implies $A\phi _i=0$ for all $i$. It follows that $A=0$. Conversely, if $0$ is an eigenvalue of $\rho$ and
$\phi$ is a corresponding unit eigenvector, then setting $P_\phi=\ket{\phi}\bra{\phi}$ we have
\begin{equation*}
\rmtr (\rho P_\phi ^*P_\phi )=\rmtr (\rho P_\phi )=\elbows{\phi ,\rho\phi}=0
\end{equation*}
But $P_\phi\ne 0$ so $\rho$ is not faithful.
\end{proof}

\begin{thm}    
\label{thm24}
If $\rho$ is faithful. then the following statements are equivalent.
{\rm{(i)}}\enspace The uncertainty inequality of Theorem~\ref{thm22}(ii) is an equality.
{\rm{(ii)}}\enspace $D_\rho (B)=\alpha D_\rho (A)$ for $\alpha\in\real$.
{\rm{(iii)}}\enspace $B=\alpha A+\beta I$ for $\alpha ,\beta\in\real$. If one of the conditions holds, then
\begin{equation}             
\label{eq25}
\sqbrac{\Delta _\rho (A,B)}^2=\ab{\rmcor _\rho (A,B)}^2=\Delta _\rho (A)\Delta )\rho (B)
\end{equation}
\end{thm}
\begin{proof}
(i)$\Rightarrow$(ii)\enspace If the uncertainty inequality is an equality, then
\begin{equation}             
\label{eq26}
\ab{\rmtr\sqbrac{\rho D_\rho (A)D_\rho (B)}}^2=\Delta _\rho(A)\Delta _\rho (B)
\end{equation}
We can rewrite \eqref{eq26} as
\begin{equation*}
\ab{\elbows{D_\rho (A),D_\rho (B)}_\rho}^2=\elbows{D_\rho (A),D_\rho (A)}_\rho\elbows{D_\rho (B),D_\rho (B)}_\rho
\end{equation*}
Since we have equality in Schwarz's inequality and $\elbows{\tbullet ,\tbullet}_\rho$ is an inner product, it follows that $D_\rho (B)=\alpha D_\rho (A)$ for some $\alpha\in\complex$. Since 
$D_\rho (B)^*=D_\rho (B)$ and $D_\rho (A)^*=D_\rho (A)$ we conclude that $\alpha\in\real$.
(ii)$\Rightarrow$(iii)\enspace If $D_\rho (B)=\alpha D_\rho (A)$ for $\alpha\in\real$, we have
\begin{equation*}
B-\elbows{B}_\rho I=\alpha\paren{A-\elbows{A}_\rho I)}
\end{equation*}
Hence, letting $\beta =\elbows{B}_\rho -\alpha\elbows{A}_\rho$ we have $B=\alpha A+\beta I$. Since $A,B\in\lscript_S(H)$ and
$\alpha\in\real$, we have that $\beta\in\real$.
(iii)$\Rightarrow$(i)\enspace If (iii) holds, then
\begin{equation*}
\elbows{B}_\rho =\rmtr (\rho B)=\alpha\rmtr (\rho A)+\beta =\alpha\elbows{A}_\rho +\beta
\end{equation*}
Hence, $\beta =\elbows{B}_\rho -\alpha\elbows{A}_\rho$ so that
\begin{align*}
D_\rho (B)&=B-\elbows{B}_\rho I=\alpha A+\beta I-\elbows{B}_\rho I\\
   &=\alpha A+\elbows{B}_\rho I-\alpha\elbows{A}_\rho I-\elbows{B}_\rho I=\alpha D_\rho (A)
\end{align*}
Thus, (ii) holds and it follows that \eqref{eq26} holds and this implies (i). Equation~\eqref{eq25} holds because \eqref{eq26} holds.
\end{proof}

\begin{exam}{2}  
The simplest faithful state when $\dim H=n<\infty$ is $\rho =I/n$. Then $\elbows{A,B}_\rho =\tfrac{1}{n}\,\rmtr (A^*B)$ which is essentially the Hilbert-Schmidt inner product $\elbows{A,B}_{HS}=\rmtr (A^*B)$. In this case for $A,B\in\lscript _S(H)$ we have
$\elbows{A}_\rho =\tfrac{1}{n}\,\rmtr (A)$, $D_\rho (A)=A-\tfrac{1}{n}\,\rmtr (A)I$. The other statistical concepts become:
\begin{align*}
\rmcor _\rho (A,B)&=\rmtr\sqbrac{\rho D_\rho (A)D_\rho (B)}=\tfrac{1}{n}\,\rmtr (AB)-\tfrac{1}{ n^2}\,\rmtr (A)\rmtr (B)\\
\Delta _\rho (A,B)&=\tfrac{1}{n}\,\rmre\rmtr (AB)-\tfrac{1}{n^2}\,\rmtr (A)\rmtr (B)\\
\Delta _\rho (A)&=\tfrac{1}{n}\,\rmtr (A^2)-\sqbrac{\tfrac{1}{n}\,\rmtr (A)}^2\\
\rmtr\paren{\rho\sqbrac{A,B}}&=\tfrac{2i}{n}\,\rmim\rmtr (AB)
\end{align*}
The uncertainty principle is given by:
\begin{align*}
\sqbrac{\rmim\rmtr (AB)}^2&+\sqbrac{\rmre\rmtr(AB)-\tfrac{1}{n}\,\rmtr (A)\rmtr (B)}^2
   =\ab{\rmtr (AB)-\tfrac{1}{n}\,\rmtr (A)\rmtr (B)}^2\\
   &\le\sqbrac{\rmtr (A^2)-\tfrac{1}{n}\,\rmtr (A)^2}\sqbrac{\rmtr (B^2)-\tfrac{1}{n}\,\rmtr (B)^2}\hskip 6.5pc\square
\end{align*}
\end{exam}

\section{Real-Valued Observables}  
An \textit{effect} is an operator $C\in\lscript _S(H)$ that satisfies $0\le C\le I$ \cite{bgl95,hz12,lud51}. Effects are thought of as two outcomes \ityes-\itno\ measurements. When the result of measuring $C$ is \ityes , we say that $C$ \textit{occurs} and when the result is \itno , then $C$ \textit{does not occur}. A \textit{real-valued observable} is a finite set of effects $A=\brac{A_x\colon x\in\Omega _A}$ where $\sum\limits _{x\in\Omega _A}A_x=I$ and $\Omega _A\subseteq\real$ is the \textit{outcome space} for $A$. The effect $A_x$ occurs when the result of measuring $A$ is the outcome $x$. The condition $\sum\limits _{x\in\Omega _A}A_x=I$ specifies that one of the possible outcomes of $A$ must occur. An observable is also called a \textit{positive operator-valued measure} (POVM). We say $A$ is \textit{sharp} if $A_x$ is a projection for all $x\in\Omega _A$ and in this case, $A$ is a \textit{projection-valued measure} 
\cite{hz12, nc00}. Corresponding to $A$ we have the \textit{stochastic operator} $\atilde\in\lscript (H)$ given by
$\atilde =\sum\limits _{x\in\Omega _A}xA_x$. Notice that we need $A$ to be real-valued in order for $\atilde$ to exist.

We now apply the theory presented in Section~2 to real-valued observables. For $\rho\in\sscript (H)$, the $\rho$-\textit{average}
(or $\rho$-\textit{expectation}) of $A$ is defined by
\begin{equation}                
\label{eq31}
\elbows{A}_\rho=\elbows{\atilde\,}_\rho=\rmtr (\rho\atilde\,)=\sum _{x\in\Omega _A}x\rmtr (\rho A_x)
\end{equation}
We interpret $\rmtr (\rho A_x)$ as the probability that a measurement of $A$ results in the outcome $x$ when the system is in state
$\rho$. Thus, \eqref{eq31} says that the $\rho$-average of $A$ is the sum of its outcomes times the probabilities these outcomes occur. We define the $\rho$-\textit{deviation} of $A$ by
\begin{align*}
D_\rho (A)&=D_\rho (\atilde\,)=\atilde -\elbows{A}_\rho I=\sum _{x\in\Omega _A}xA_x-\sum _{x\in\Omega _A}x\rmtr (\rho A_x)I\\
   &=\sum _{x\in\Omega _A}x\sqbrac{A_x-\rmtr (\rho A_x)I}
\end{align*}
If $A,B$ are real-valued observables, the $\rho$-\textit{correlation} of $A,B$ is $\rmcor _\rho (A,B)=\rmcor _\rho (\atilde ,\btilde\,)$,
$\rho$-\textit{covariance} of $A,B$ is $\Delta _\rho (A,B)=\Delta _\rho (\atilde ,\btilde\,)$ and the $\rho$-\textit{variance} of $A$ is
$\Delta _\rho (A)=\Delta _\rho (\atilde\,)$. Applying \eqref{eq21} we obtain
\begin{align}                
\label{eq32}
\rmcor _\rho (A,B)&=\rmtr (\rho\atilde\btilde\,)-\elbows{\atilde\,}_\rho\elbows{\btilde\,}_\rho
   =\rmtr\paren{\rho\sum _{x,y}xyA_xB_y}-\elbows{\atilde\,}_\rho\elbows{\btilde\,}_\rho\notag\\
   &=\sum _{x,y}xy\sqbrac{\rmtr (\rho A_xB_y)-\rmtr (\rho A_x)\rmtr (\rho B_y)}
\end{align}
It follows that
\begin{align}                
\label{eq33}
\Delta _\rho (A,B)&=\sum _{x,y}xy\sqbrac{\rmre\rmtr (\rho A_xB_y)-\rmtr (\rho A_x)\rmtr (\rho B_y)}\\
\intertext{and}
\label{eq34}
\Delta _\rho (A)&=\sum _{x,y}xy\sqbrac{\rmtr (\rho A_xA_y)-\rmtr (\rho A_x)\rmtr (\rho A_y)}
\end{align}
We also have by \eqref{eq24} that
\begin{align}                
\label{eq35}
\rmtr\paren{\rho\sqbrac{\atilde ,\btilde\,}}&=2i\,\rmim\rmtr (\rho\atilde\btilde\,)=2i\,\rmim\rmtr\paren{\rho\sum _{x,y}xyA_xB_y)}\notag\\
   &=2i\sum _{x,y}xy\,\rmim\rmtr (\rho A_xB_y)
\end{align}
Substituting $\atilde,\btilde$ for $A,B$ in Theorem~\ref{thm22} gives an uncertainty principle for real-valued observables.

Two observables $A,B$ are \textit{compatible} (or \textit{jointly measurable}) if there exists a \textit{joint observable}
$C_{(x,y)}$, $(x,y)\in\Omega _a\times\Omega _B$, such that $A_x=\sum\limits _yC_{(x,y)}$, $B_y=\sum\limits _xC_{(x,y)}$ for all
$x\in\Omega _a$, $y\in\Omega _B$. If $\sqbrac{A_x,B_y}=0$ for all $x,y$, then $A,B$ are compatible with $C_{(x,y)}=A_xB_y$ for all $(x,y)\in\Omega _A\times\Omega _B$. However, if $A,B$ are compatible, they need not commute \cite{hz12}. If $A,B$ are compatible real-valued observables, then
\begin{align*}
\atilde&=\sum _xxA_x=\sum _{x,y}xC_{(x,y)}\\
\btilde&=\sum _yyB_y=\sum _{x,y}yC_{(x,y)}
\end{align*}
Using \eqref{eq32}, \eqref{eq33}, \eqref{eq34}, \eqref{eq35} we can write
$\rmcor _\rho (A,B),\Delta _\rho (A,B),\Delta _\rho (A),\Delta _\rho (B)$ and $\rmtr\paren{\rho\sqbrac{\atilde ,\btilde\,}}$ in terms of
$C_{(x,y)}$. Hence, we can express the uncertainty principle in terms of $C_{(x,y)}$.

If $A=\brac{A_x\colon x\in\Omega _a}$ is a real-valued observable, then $\atilde$ has spectral decomposition
$\atilde =\sum\limits _{i=1}^n\lambda _iP_i$ where $\lambda _i\in\real$ are the distinct eigenvalues of $\atilde$ and $P_i$ are projections with $\sum P_i=I$. We call $\ahat =\brac{P_i\colon i=1,2,\ldots ,n}$ the \textit{sharp version} of $A$. Then $\ahat$ is a real-valued observable with outcome space $\Omega _{\ahat}=\brac{\lambda _i\colon i=1,2,\ldots ,n}$. Since $(\ahat\,)^\sim =\atilde$, $A$ and $\ahat$ have the same stochastic operator. It follows that $\elbows{A}_\rho=\elbows{\ahat\,}_\rho$,
$\Delta _\rho (A)=\Delta _\rho (\ahat\,)$ and if $B$ is another real-valued observable, then
$\rmcor _\rho (A,B)=\rmcor _\rho (\ahat ,\bhat\,)$ and $\Delta _\rho (A,B)=\Delta _\rho (\ahat ,\bhat\,)$.

\begin{lem}    
\label{lem31}
The following statements are equivalent.
{\rm{(i)}}\enspace $\ahat =\bhat$.
{\rm{(ii)}}\enspace $\atilde =\btilde$.
{\rm{(iii)}}\enspace $\elbows{A}_\rho=\elbows{B}_\rho$ for all $\rho\in\sscript (H)$.
\end{lem}
\begin{proof}
(i)$\Rightarrow$(ii)\enspace If $\ahat =\bhat$ then 
\begin{equation*}
\atilde =(\ahat\,)^\sim=(\bhat\,)^\sim =\btilde
\end{equation*}
(ii)$\Rightarrow$(iii)\enspace If $\atilde =\btilde$ then
\begin{equation*}
\elbows{A}_\rho =\elbows{\atilde\,}_\rho =\elbows{\btilde\,}_\rho =\elbows{B}_\rho
\end{equation*}
(iii)$\Rightarrow$(i)\enspace If $\elbows{A}_\rho =\elbows{B}_\rho$ for all $\rho\in\sscript (H)$, then
$\elbows{\atilde\,}_\rho =\elbows{\btilde\,}_\rho$ for all $\rho\in\sscript (H)$. It follows that $\ahat =\bhat$.
\end{proof}

Let $\atilde =\sum xA_x=\sum\lambda _iP_i$ so $\ahat =\brac{P_i\colon i=1,2,\ldots ,n}$ is a sharp version of $A$. Let
$B=\brac{B_x\colon x\in\Omega _A}$ be the real-valued observable given by $B_x=\sum\limits _{i=1}^nP_iA_xP_i$. We conclude that $A$ and $B$ have the same sharp version because
\begin{align*}
\btilde&=\sum _xxB_x=\sum _iP_i\sum _xxA_xP_i=\sum _iP_i\atilde P_i=\sum _iP_i\sum _j\lambda _jP_jP_i\\
   &=\sum _{i,j}\lambda _iP_iP_jP_i=\sum _i\lambda _iP_i=\atilde
\end{align*}
so by Lemma~\ref{lem31}, $\ahat =\bhat$. We say that $B$ is a \textit{conjugate} of $A$. Letting $C_{ix}=P_iA_xP_i$, we have that
\begin{equation*}
\brac{C_{ix}\colon i=1,2,\ldots ,n, x\in\Omega _A}
\end{equation*}
is an observable and $\sum\limits _iC_{ix}=B_x$, $\sum\limits _xC_{ix}=P_i$. It follows that $B$ and $\ahat$ are compatible with joint observable $\brac{C_{ix}}$. We say that an observable $A=\brac{A_x\colon x\in\Omega _A}$ is \textit{commutative} if
$\sqbrac{A_x,A_y}=0$ for all $x,y\in\Omega _A$. Notice that if $A$ is sharp, then $A$ is commutative. However, there are many unsharp observables that are commutative.

\begin{thm}    
\label{thm32}
If $A$ is commutative, then $B$ is conjugate to $A$ if and only if $B=A$.
\end{thm}
\begin{proof}
If $A$ is commutative, we show that $A$ is conjugate to $A$. Since
\begin{equation*}
\ahat =\sum xA_x=\sum\lambda _iP_i
\end{equation*}
we have that $\sqbrac{\ahat ,A_x}=0$ for all $x\in\Omega _A$. By the spectral theorem, $\sqbrac{A_x,P_i}=0$ for all $x,i$ so
$A_x=\sum P_iA_xP_i$. Therefore, $A$ is conjugate to $A$. Conversely, suppose $A$ is commutative and $B$ is conjugate to $A$. Then $B_x=\sum\limits _iP_iA_xP_i$ for all $x\in\Omega _A$. As before, we have that $\sqbrac{\ahat _x,A_x}=0$ for all
$x\in\Omega _A$ so $\sqbrac{A_x,P_i}=0$ for all $x,i$. Hence,
\begin{equation*}
B_x=\sum _iP_iA_xP_i=A_x\sum _iP_i=A_x
\end{equation*}
for all $x\in\Omega _B=\Omega _A$ so $B=A$.
\end{proof}
Thus, nontrivial conjugates only occur in the nonclassical case where $A$ is noncommutative.

\section{More Examples}  
This section illustrates the theory in Sections~2 and 3 with two examples.

\begin{exam}{3}  
A two outcome observable is called a \textit{dichotomic observable}. Of course, a dichotomic observable is commutative but it need not be sharp. Let $A=\brac{A_1,I-A_1}$ be a dichotomic observable with $\Omega _A=\brac{1,-1}$. Then
\begin{align*}
\atilde&=A_1-(I-A_1)=2A_1-I\\
\elbows{A}_\rho&=\rmtr (\rho\atilde\,)=\rmtr\sqbrac{\rho (2A_1-I)}=2\,\rmtr (\rho A_1)-1\\
D_\rho (A)&=\atilde -\elbows{A}_\rho I=2A_1-I-2\,\rmtr (\rho A_1)I+I=2\sqbrac{A_1-\rmtr (\rho A_1)I}
\end{align*}
If $B=\brac{B_1,I-B_1}$ is another dichotomic observable with $\Omega _B=\brac{1,-1}$, then
\begin{align}      
\label{eq41}
\rmcor _\rho (A,B)&=\rmtr (\rho\atilde\btilde\,)-\elbows{A}_\rho\elbows{B}_\rho\notag\\
   &=\rmtr\sqbrac{\rho (2A_1-I)(2B_1-I)}-\sqbrac{2\,\rmtr (\rho A_1-1)}\sqbrac{2\,\rmtr (\rho B_1-1)}\notag\\
   &=\rmtr\sqbrac{\rho (4A_1B_1-2A_1-2B_1+I)}-4\,\rmtr (\rho A_1)\rmtr (\rho B_1)\notag\\
   &\quad +2\,\rmtr (\rho A_1)+2\rmtr (\rho B_1)-1\notag\\
   &=4\sqbrac{\rmtr (\rho A_1B_1)-\rmtr (\rho A_1)\rmtr (\rho B_1)}
\end{align}
Hence,
\begin{align*}
\Delta _\rho (A,B)&=4\sqbrac{\rmre\rmtr (\rho A_1B_1)-\rmtr (\rho A_1)\rmtr (\rho B_1)}\\
\intertext{and}
\Delta _\rho (A)&=\Delta _\rho (A,A)=4\sqbrac{\rmtr (\rho A_1^2)-\paren{\rmtr (\rho A_1)}^2}
\end{align*}
We also have
\begin{align*}
\sqbrac{\atilde ,\btilde\,}&=\sqbrac{2A_1-I,2B_1-I}=(2A_1-I)(2B_1-I)-(2B_1-I)(2A_1-I)\\
   &=4\sqbrac{A_1,B_1}
\end{align*}
We conclude that $\sqbrac{\atilde ,\btilde\,}=0$ if and only if $\sqbrac{A_1,B_1}=0$ and this does not hold in general so
$\atilde,\btilde$ need not commute. The uncertainty principle becomes
\begin{align}      
\label{eq42}
&\sqbrac{\rmim\rmtr (\rho A_1B_1)}^2+\sqbrac{\rmre\rmtr (\rho A_1B_1)-\rmtr (\rho A_1)\rmtr (\rho A_2)}^2\notag\\
   &\hskip 2pc =\ab{\rmtr (\rho A_1B_1)-\rmtr (\rho A_1)\rmtr (\rho B_1)}^2\notag\\
   &\hskip 2pc \le\sqbrac{\rmtr (\rho A_1^2)-\paren{\rmtr (\rho A_1)}^2}\sqbrac{\rmtr (\rho B_1^2)-\paren{\rmtr (\rho B_1)}^2}\quad\square
\end{align}
\end{exam}
 
 \begin{exam}{4}  
We now consider a special case of Example~3. For $H\in\complex ^2$ we define the Pauli matrices
\begin{equation*}
\sigma _x=\begin{bmatrix}0&1\\1&0\end{bmatrix},\quad
   \sigma _y=\begin{bmatrix}0&i\\-i&0\end{bmatrix},\quad
  \sigma _z=\begin{bmatrix}1&0\\0&-1\end{bmatrix} 
\end{equation*}
Let $\mu\in\sqbrac{0,1}$ and define the dichotomic observable $A=\brac{A_1,I-A_1}$, where
\begin{equation*}
A_1=\tfrac{1}{2}(I+\mu\sigma _x)=\tfrac{1}{2}\begin{bmatrix}1&\mu\\\mu&1\end{bmatrix}
\end{equation*}
and $\Omega _A=\brac{1,-1}$. Similarly, let $B=\brac{B_1,I-B_1}$, where
\begin{equation*}
B_1=\tfrac{1}{2}(I+\mu\sigma _y)=\tfrac{1}{2}\begin{bmatrix}1&i\mu\\-i\mu&1\end{bmatrix}
\end{equation*}
and $\Omega _B=\brac{1,-1}$. We call $A$ and $B$ \textit{noisy spin observables} along the $x$ and $y$ directions, respectively, with
\textit{noise parameter} $1-\mu$ \cite{nc00}.

\parindent 20pt Any state $\rho\in\sscript (H)$ has the form $\rho =\tfrac{I}{2}(I+\rarrow\tbullet\sigmaarrow )$ where $\rarrow\in\real ^3$ with
$\doubleab{\rarrow}\le1$ \cite{bgl95,gud120}. This is called the \textit{Block sphere} representation of $\rho$ \cite{hz12,nc00}.
The eigenvalues of $\rho$ are $\lambda _{\pm} =\tfrac{1}{2}\paren{1\pm\doubleab{\rarrow}}$. Then $\lambda _+=1$, $\lambda _-=0$ if and only if $\doubleab{\rarrow}=1$ and these are precisely the pure states. Letting $\sigma _1=\sigma _x$, $\sigma _2=\sigma _y$,
$\sigma _3=\sigma _z$ we obtain
\begin{align*}
\rho &=\frac{1}{2}\begin{bmatrix}1+r_3&r_1-ir_2\\r_1+ir_2&1-r_3\end{bmatrix}\\
\intertext{and}
\rho A_1&=\frac{1}{4}\begin{bmatrix}1+r_3&r_1-ir_2\\r_1+ir_2&1-r_3\end{bmatrix}\begin{bmatrix}1&\mu\\\mu&1\end{bmatrix}\\
&=\begin{bmatrix}1+r_3+(r_1-ir_2)\mu&(1+r_3)\mu +r_1-ir_2\\(1-r_3)\mu +r _1+ir_2&1-r_3+(r_1+ir_2)\mu\end{bmatrix}
\end{align*}
Hence, $\rmtr (\rho A_1)=\tfrac{1}{2}(1+r_1\mu )$ and as in Example~3, $\elbows{A}_\rho =r_1\mu$. Similarly,
$\rmtr (\rho B_1)=\tfrac{1}{2}(1+r_2\mu )$ and $\elbows{B}_\rho=r_2\mu$. We also obtain
\begin{equation*}
\rmtr (\rho A_1B_1)=\tfrac{1}{4}\sqbrac{1+(r_1+r_2)\mu +ir_2\mu ^2}
\end{equation*}
and it follows from \eqref{eq41} that
\begin{align*}
\rmcor _\rho (A,B)&=4\sqbrac{\rmtr (\rho A_1B_1)-\rmtr (\rho A_1)\rmtr (\rho B_1)}\\
   &=1+(r_1+r_2)\mu +ir_3\mu ^2-(1+r_1\mu )(1+r_2\mu)=-r_1r_2\mu ^2+ir_3\mu ^2
\end{align*}
Therefore, $\Delta _\rho (A,B)=-r_1r_2\mu ^2$. A straightforward calculation shows that
\begin{align*}
\rmtr (\rho A_1^2)&=\tfrac{1}{4}(1+\mu ^2)+\tfrac{1}{2}\mu r_1\\
\rmtr (\rho B_1^2)&=\tfrac{1}{4}(1+\mu ^2)+\tfrac{1}{2}\mu r_2
\end{align*}
It follows that
\begin{equation*}
\Delta _\rho (A)=4\sqbrac{\rmtr (\rho A_1^2)-\paren{\rmtr (\rho A_1)}^2}=\mu ^2(1-r_1^2)
\end{equation*}
and similarly, $\Delta _\rho (B)=\mu ^2(1-r_2^2)$.

The commutator term in \eqref{eq42} becomes
\begin{equation*}
\sqbrac{\rmim\rmtr (\rho A_1B_1)}^2=\tfrac{1}{16}\,r_3^2\mu ^4
\end{equation*}
The covariance term in \eqref{eq42} is
\begin{equation*}
\sqbrac{\rmre (\rho A_1B_1)-\rmtr (\rho A_1)\rmtr (\rho B_1)}^2=\tfrac{1}{16}\,r_1^2r_2^2\mu ^4
\end{equation*}
and the correlation term in \eqref{eq42} is
\begin{equation*}
\ab{\rmtr (\rho A_1B_1)-\rmtr (\rho A_1)\rmtr (\rho B_1)}^2=\tfrac{1}{16}\,(r_3^2+r_1^2r_2^2)\mu ^4
\end{equation*}
Finally, the variance term in \eqref{eq42} is given by
\begin{equation*}
\Delta _\rho (A_1)\Delta _\rho (B_1)=\tfrac{1}{16}\,(1-r_1^2)(1-r_2^2)\mu ^4
\end{equation*}
The inequality in \eqref{eq42} reduces to
\begin{equation}                
\label{eq43}
\tfrac{1}{16}(r_3^2+r_1^2+r_2^2)\mu ^4\le\tfrac{1}{16}(1-r_1^2)(1-r_2^2)\mu ^4
\end{equation}
If $\mu\ne 0$, \eqref{eq43} is equivalent to the inequality
\begin{equation*}
\doubleab{\rarrow}^2=r_1^2+r_2^2+r_3^2\le 1
\end{equation*}
If the commutator term vanishes and $\mu\ne 0$, the uncertainty inequality becomes
\begin{equation}                
\label{eq44}
r_1^2r_2^2\le (1-r_1^2)(1-r_2^2)
\end{equation}
which is equivalent to $r_1^2+r_2^2\le 1$. If $A$ and $B$ are $\rho$-uncorrelated and $\mu\ne 0$, the uncertainty inequality becomes $r_3^2\le (1-r_1^2)(1-r_2^2)$ which is equivalent to $\doubleab{\rarrow}^2\le 1+r_1^2r_2^2$. This inequality and \eqref{eq44} are weaker than \eqref{eq43}.\hfill\qedsymbol
\end{exam}

\section{Real-Valued Coarse Graining}  
Let $A=\brac{A_x\colon x\in\Omega _A}$ be an arbitrary observable. We assume that $A$ is not necessarily real-valued so the outcome space $\Omega _A$ is an arbitrary finite set. For $f\colon\Omega _A\to\real$ with range $\rscript (f)$ we define the real-valued observable $f(A)$ by $\Omega _{f(A)}=\rscript (f)$ and for all $z\in\Omega _{f(A)}$
\begin{equation*}
f(A)_z=A_{f^{-1}(z)}=\sum\brac{A_x\colon f(x)=z}
\end{equation*}
We call $f(A)$ a \textit{real-valued coarse graining} of $A$ \cite{gud120,gud220,hz12}. Then $f(A)$ has stochastic operator
\begin{equation*}
f(A)^\sim=\sum _zzf(A)_z=\sum _zzA_{f^{-1}(z)}=\sum _z\sum _{x\in f^{-1}(z)}zA_x=\sum _xf(x)A_x
\end{equation*}
It follows that $\elbows{f(A)}_\rho =\sum\limits _xf(x)\rmtr (\rho A_x)$ for all $\rho\in\sscript (H)$. If $B$ is another observable and $g\colon\Omega _B\to\real$ we have
\begin{align*}
\rmcor _\rho\sqbrac{f(A),g(B)}&=\sum _{x,y}f(x)g(y)\rmtr (\rho A_xB_y)-\elbows{f(A)}_\rho\elbows{g(B)}_\rho\\
\Delta _\rho\sqbrac{f(A),g(B)}&=\sum _{x,y}f(x)g(y)\rmre\rmtr (\rho A_xB_y)-\elbows{f(A)}_\rho\elbows{g(B)}_\rho\\
\Delta _\rho\sqbrac{f(A)}&=\sum _{x,y}f(x)f(y)\rmtr (\rho A_xA_y)-\elbows{f(A)}_\rho ^2
\end{align*}
Moreover, we have the uncertainty inequality
\begin{equation*}
\ab{\rmcor _\rho\sqbrac{f(A),g(B)}}^2\le\Delta _\rho\sqbrac{f(A)}\Delta _\rho\sqbrac{g(B)}
\end{equation*}
We denote the set of trace-class operators on $H$ by $\tscript (H)$. An \textit{operation} on $H$ is a completely positive, trace reducing, linear map $\oscript\colon\tscript (H)\to\tscript (H)$ \cite{bgl95,gud120,gud220,hz12}. If $\oscript$ preserves the trace, then
$\oscript$ is called a \textit{channel}. A (finite) \textit{instrument} is a finite set of operators
$\iscript =\brac{\iscript _x\colon x\in\Omega _\iscript}$ such that $\iscriptbar =\sum\brac{\iscript _x\colon x\in\Omega _\iscript}$ is a channel \cite{bgl95,gud120,gud220,hz12}. We say that $\iscript$ \textit{measures } an observable $A$ if
$\Omega _\iscript =\Omega _A$ and $\rmtr\sqbrac{\iscript _x(\rho )}=\rmtr (\rho A_x)$ for all $x\in\Omega _\iscript$. It can be shown that $\iscript$ measures a unique observable which we denote by $J(\iscript )$ \cite{gud120,gud220}. Conversely, any observable is measured by many instruments \cite{bgl95,gud120,gud220,hz12}. Corresponding to an operation $\oscript$ we have its
\textit{dual-operation} $\oscript ^*\colon\lscript (H)\to\lscript (H)$ defined by
$\rmtr\sqbrac{\rho\oscript ^*(C)}=\rmtr\sqbrac{\oscript (\rho )C}$ for all $\rho\in\sscript (H)$ \cite{gud120,gud220}. It can be shown that $J(\iscript )_x=\iscript _x^*(I)$ for all $x\in\Omega _\iscript$ where $I$ is the identity operator \cite{gud120,gud220}.

As with observables, if $\iscript$ is an instrument, and $f\colon\Omega _\iscript\to\real$ we define the real-valued instrument $f(\iscript )$ such that $\Omega _{f(\iscript )}=\rscript (f)$ and
\begin{equation*}
f(\iscript )_z=\sum\brac{\iscript _x\colon f(x)=z}
\end{equation*}
If $J(\iscript )=A$, then $J\sqbrac{f(\iscript )}=f(A)$ because
\begin{align*}
\rmtr\sqbrac{f(\iscript )_z(\rho )}&=\rmtr\sqbrac{\sum\brac{\iscript _x(\rho )\colon f(x)=z}}
   =\sum\brac{\rmtr\sqbrac{\iscript _x(\rho )}\colon f(x)=z}\\
   &=\sum\brac{\rmtr (\rho A_x)\colon f(x)=z}=\rmtr\sqbrac{\rho\sum\brac{A_x\colon f(x)=z}}\\
   &=\rmtr\sqbrac{\rho f(A)_z}
\end{align*}
for all $z\in\Omega _{f(A)}=\Omega _{f(\iscript )}$. If $\iscript$ is real-valued, we define $\iscripttilde$ on $\lscript (H)$ by
$\iscripttilde(C)=\sum x\iscript _x(C)$ and $\elbows{\iscript}_\rho =\rmtr\sqbrac{\iscripttilde (\rho )}$. If $J(\iscript )=A$, then
\begin{equation*}
\elbows{\iscript}_\rho=\rmtr\sqbrac{\sum x\iscript _x(\rho )}=\sum x\rmtr\!\sqbrac{\iscript _x(\rho )}=\sum x\rmtr (\rho A_x)
   =\elbows{A}_\rho
\end{equation*}
for all $\rho\in\sscript (H)$. We also define $\Delta _\rho (\iscript )=\Delta _\rho (A)$. It follows that 
$\elbows{f(\iscript )}_\rho =\elbows{f(A)}_\rho$, $\Delta _\rho\sqbrac{f(\iscript )}=\Delta _\rho\sqbrac{f(A)}$ and
$f(\iscript )^\sim=\sum f(x)\iscript _x$.

Let $A=\brac{A_x\colon x\in\Omega _A}$, $B=\brac{B_y\colon y\in\Omega _B}$ be arbitrary observables and suppose $\iscript$ is an instrument with $J(\iscript )=A$. Define the $\iscript$-\textit{product} observable $A\circ B$ with
$\Omega _{A\circ B}=\Omega _A\times\Omega _B$ given by $(A\circ B)_{(x,y)}=\iscript _x(B_y)$ \cite{gud120,gud220}. Then
$A\circ B$ is indeed an observable because
\begin{equation*}
\sum _{x,y}(A\circ B)_{(x,y)}=\sum _{x,y}\iscript _x^*(B_y)=\sum _x\iscript _x^*\paren{\sum _yB_y}
   =\sum _x\iscript _x^*(I)=\sum _xA_x=I
\end{equation*}
Although $A\circ B$ depends on $\iscript$, we shall not indicate this for simplicity. We interpret $A\circ B$ as the observable obtained by first measuring $A$ using $\iscript$ and then measuring $B$. If $f\colon\Omega _A\times\Omega _B\to\real$ we obtain the real-valued observable $f(A,B)=f(A\circ B)$. We then have
\begin{align*}
f(A,B)_z&=(A\circ B)_{f^{-1}(z)}=\sum\brac{(A\circ B)_{(x,y)}\colon f(x,y)=z}\\
   &=\sum\brac{\iscript _x^*(B_y)\colon f(x,y)=z}\\
  f(A,B)^\sim&=\sum _{x,y}f(x,y)(A\circ B)_{(x,y)}=\sum_{x,y}f(x,y)\iscript_x^*(B_y)\\
   \elbows{f(A,B)}_\rho&=\sum _{x,y}f(x,y)\rmtr\sqbrac{\rho (A\circ B)_{(x,y)}}
   =\sum _{x,y}f(x,y)\rmtr\sqbrac{\rho\iscript _x^*(B_y)}\\
   \Delta _\rho\sqbrac{f(A,B)}
   &=\sum _{x,y,x',y'}f(x,y)f(x',y')\rmtr\sqbrac{\rho (A\circ B)_{(x,y)}(A\circ B)_{(x',y')}}-\elbows{f(A,B)}_\rho ^2\\
   &=\rmtr\brac{\rho\sqbrac{\sum _{x,y}f(x,y)\iscript _x^*(B_y)}^2}-\elbows{f(A,B)}_\rho ^2
\end{align*}
If $f$ is a product function $f(x,y)=g(x)h(y)$ we obtain
\begin{equation*}
f(A,B)_z=\sum _z\brac{\iscript _x^*(B_y)\colon g(x)h(y)=z}
\end{equation*}
We then have the simplification
\begin{align*}
f(A,B)^\sim&=\sum _{x,y}g(x)h(y)\iscript _x^*(B_y)=\sum _xg_x\iscript _x^*\paren{\sum _yh(y)B_y}\\
   &=\sum _xg(x)\iscript _x^*\sqbrac{h(B)^\sim}\\
\end{align*}
Hence,
\begin{align*}
\elbows{f(A,B)}_\rho&=\rmtr\sqbrac{\rho f(A,B)^\sim}=\rmtr\brac{\rho\sum _xg(x)\iscript _x^*\sqbrac{h(B)^\sim}}\\
   &=\sum _xg(x)\rmtr\brac{\rho\iscript _x^*\sqbrac{h(B)^\sim}}=\sum _xg(x)\rmtr\brac{\iscript _x(\rho )\sqbrac{h(B)^\sim}}\\
   &=\rmtr\brac{\sum _xg(x)\iscript _x(\rho )\sqbrac{h(B)^\sim}}=\rmtr\brac{g(\iscript )^\sim (\rho )\sqbrac{h(B)^\sim}}
\end{align*}
In a similar way we obtain
\begin{equation*}
\Delta _\rho\sqbrac{f(A,B)}=\rmtr\brac{\paren{g(\iscript )^\sim (\rho )\sqbrac{h(B)^\sim}}^2}-\elbows{f(A,B)}_\rho ^2
\end{equation*}

If $A$ and $B$ are arbitrary observables, we define the observable $B$ \textit{conditioned} by $A$ to be
\begin{equation*}
(B\mid A)_y=\iscript _{\Omega _A}^*(B_y)=\sum _{x\in\Omega _A}\iscript _x^*(B_y)
\end{equation*}
where $\Omega _{B\mid A}=\Omega _B$ \cite{gud120,gud220}. We interpret $(B\mid A)$ as the observable obtained by first measuring $A$ without taking the outcome into account and then measuring $B$. If $B$ is real-valued we have
\begin{align*}
(B\mid A)^\sim&=\sum _yy(B\mid A)_y=\sum _{x,y}y\iscript _x^*(B_y)=\iscript _{\Omega (A)}^*(\btilde )\\
   \elbows{(B\mid A)}_\rho&=\sum _yy\rmtr\sqbrac{\rho\iscript _{\Omega (A)}^*(B_y)}=\sum y\rmtr\sqbrac{\iscriptbar (\rho )B_y}
   =\rmtr\sqbrac{\iscriptbar (\rho )\btilde\,}=\elbows{B}_{\iscriptbar _(\rho )}\\
   \Delta _\rho\sqbrac{(B\mid A)}&=\Delta _\rho\sqbrac{(B\mid A)^\sim}=\Delta _\rho\sqbrac{\iscript _{\Omega (A)}^*(\btilde )}
   =\rmtr\brac{\sqbrac{\iscript _{\Omega (A)}^*(\btilde )}^2}-\sqbrac{\elbows{B}_{\iscriptbar (\rho )}}^2
\end{align*}

We now illustrate the theory of this section with some examples.

\begin{exam}{5}  
The simplest example of an instrument is a \textit{trivial instrument} $\iscript _x(\rho )=\omega (x)\rho$ where $\omega$ is a probability measure on the finite set $\Omega _\iscript$. It is clear that $\iscript$ measures the \textit{trivial observable} $A_x=\omega (x)I$. Let $B$ be an arbitrary observable and let $f\colon\Omega _A\times\Omega _B\to\real$. We then have
\begin{align*}
(A\circ B)_{(x,y)}&=\iscript _x^*(B_y)=\omega (x)B_y\\
f(A,B)_z&=f(A\circ B)_z=\sum\brac{\omega (x)B_y\colon f(x,y)=z}
\end{align*}
We conclude that
\begin{align*}
f(A,B)^\sim&=\sum _{x,y}f(x,y)\omega (x)B_y\\
\elbows{f(A,B)}_\rho&=\sum _{x,y}f(x,y)\omega (x)\rmtr (\rho B_y)\\
\Delta _\rho\sqbrac{f(A,B))}&=\rmtr\brac{\rho\sqbrac{\sum _{x,y}f(x,y)\omega (x)B_y}^2}-\elbows{f(A,B)}_\rho ^2
\end{align*}
Moreover, since
\begin{equation*}
(B\mid A)_y=\sum _x\iscript _x^*(B_y)=\sum _x\omega (x)(B_y)=B_y
\end{equation*}
we have that $(B\mid A)=B$.\hfill\qedsymbol
\end{exam}

\begin{exam}{6}  
Let $A=\brac{A_x\colon x\in\Omega _A}$ and $B=\brac{B_y\colon y\in\Omega _B}$ be arbitrary observables and let
$\hscript _x(\rho )=\rmtr (\rho A_x)\alpha _x$, $\alpha _x\in\sscript (H)$ be a \textit{Holevo instrument} \cite{gud120,gud220}. Then
$\hscript$ measure $A$ because
\begin{equation*}
\rmtr\sqbrac{\hscript _x(\rho )}=\rmtr\sqbrac{\rmtr (\rho A_x)\alpha _x}=\rmtr (\rho A_x)
\end{equation*}
Since $\hscript _x^*(a)=\rmtr (\alpha _xa)A_x$ for all $x\in\Omega _A$ \cite{gud120,gud220}, we have
\begin{equation*}
(A\circ B)_{(x,y)}=\hscript _x^*(B_y)=\rmtr (\alpha _xB_y)A_x
\end{equation*}
If $f\colon\Omega _A\times\Omega _B\to\real$, we obtain the real-valued observable
\begin{equation*}
f(A,B)_z=\sum\brac{\rmtr (\alpha _xB_y)A_x\colon f(x,y)=z}
\end{equation*}
We conclude that 
\begin{align*}
f(A,B)_z&=\sum _{x,y}f(x,y)\hscript _x^*(B_y)=\sum _{x,y}f(x,y)\rmtr (\alpha _xB_y)A_x\\
   \elbows{f(A,B)}_\rho&=\sum _{x,y}f(x,y)\rmtr (\alpha _xB_y)\rmtr (\rho A_x)\\
   \Delta _\rho\sqbrac{f(A,B)}&=\sum _{x,y,x',y'}f(x,y)f(x',y')\rmtr\sqbrac{\rho\rmtr (\alpha _xB_y)A_x\rmtr (\alpha _{x'}B_{y'})A_{x'}}\\
    &\quad -\elbows{f(A,B)}_\rho ^2\\
    &=\rmtr\brac{\rho\sqbrac{\sum _{x,y}f(x,y)\rmtr (\alpha _xB_y)A_x}^2}-\elbows{f(A,B)}_\rho ^2
\end{align*}
Moreover, we have
\begin{equation*}
(B\mid A)_y=\sum _x\hscript _x^*(B_y)=\sum _x\rmtr (\alpha _xB_y)A_x\hskip 10pc\square
\end{equation*}
\end{exam}

\begin{exam}{7}  
Let $A,B$ be arbitrary observables and let $\lscript$ be the \textit{L\"uders instrument} given by
$\lscript _x(\rho )=A_x^{1/2}\rho A_x^{1/2}$ \cite{gud120,gud220,lud51}. Then
\begin{equation*}
\rmtr\sqbrac{\lscript _x(\rho )}=\rmtr (A_x^{1/2}\rho A_x^{1/2})=\rmtr (\rho A_x)
\end{equation*}
so $\lscript$ measures $A$. Since $\lscript _x^*(a)=A_x^{1/2}aA_x^{1/2}$ \cite{gud120,gud220} we have
\begin{equation*}
(A\circ B)_{(x,y)}=A_x^{1/2}B_yA_x^{1/2}
\end{equation*}
If $f\colon\Omega _A\times\Omega _B\to\real$, we obtain the real-valued observable
\begin{equation*}
f(A,B)_z=\sum\brac{A_x^{1/2}B_yA_x^{1/2}\colon f(x,y)=z}
\end{equation*}
We conclude that
\begin{align*}
f(A,B)^\sim&=\sum _{x,y}f(x,y)A_x^{1/2}B_yA_x^{1/2}\\
   \elbows{f(A,B)}_\rho&=\sum _{x,y}f(x,y)\rmtr (\rho A_x^{1/2}B_yA_x^{1/2})=\sum _{x,y}f(x,y)\rmtr (A_x^{1/2}\rho A_x^{1/2}B_y)\\
   \Delta _\rho\sqbrac{f(A,B)}&=\rmtr\brac{\rho\sqbrac{\sum _{x,y}f(x,y)A_x^{1/2}B_yA_x^{1/2}}^2}-\elbows{f(A,B)}_\rho ^2
\end{align*}
Moreover, we have
\begin{equation*}
(B\mid A)_y=\sum _x\lscript _x^*(B_y)=\sum _xA_x^{1/2}B_yA_x^{1/2}\hskip 10pc\square
\end{equation*}
\end{exam}

\end{document}